\documentclass{article}
\usepackage{macros}
\usepackage[margin=1.25in]{geometry}

\newclass{\BCH}{BCH}
\newclass{\Enc}{Enc}
\newclass{\Dec}{Dec}
\newclass{\mem}{mem}
\newclass{\val}{val}
\newclass{\bin}{bin}

\title{Fully Characterizing Lossy Catalytic Computation}
\author{Marten Folkertsma\thanks{Supported by the Dutch Ministry of Economic Affairs and Climate Policy (EZK), as part of the Quantum Delta NL program.} \\ CWI \& Qusoft \\ \texttt{mjf@cwi.nl} \and Ian Mertz\thanks{Supported by Royal Society University Research Fellowship URF\textbackslash R1\textbackslash 191059.} \\ University of Warwick \\ \texttt{ian.mertz@warwick.ac.uk} \and Florian Speelman\thanks{Supported by the Dutch Ministry of Economic Affairs and Climate Policy (EZK), as part of the Quantum Delta NL program, and the project Divide and Quantum `D\&Q' NWA.1389.20.241 of the program `NWA-ORC', which is partly funded by the Dutch Research Council (NWO).} \\ University of Amsterdam \\ \texttt{f.speelman@uva.nl} \and Quinten Tupker\thanks{Supported by the Dutch National Growth Fund (NGF), as part of the Quantum Delta NL program.} \\ CWI \\ \texttt{qmt@cwi.nl}}
\date{\today}

\begin{document}

\maketitle

\begin{abstract}
A \textit{catalytic machine} is a model of computation where a traditional space-bounded machine is augmented with an additional, significantly larger, ``catalytic'' tape, which, while being available as a work tape, has the caveat of being initialized with an arbitrary string, which must be preserved at the end of the computation. Despite this restriction, catalytic machines have been shown to have surprising additional power; a logspace machine with a polynomial length catalytic tape, known as \textit{catalytic logspace} ($\CL$), can compute problems which are believed to be impossible for $\L$.

A fundamental question of the model is whether the catalytic condition, of leaving the catalytic tape in its exact original configuration, is robust to minor deviations. This study was initialized by Gupta et al. (2024), who defined \textit{lossy catalytic logspace} ($\LCL[e]$) as a variant of $\CL$ where we allow up to $e$ errors when resetting the catalytic tape. They showed that $\LCL[e] = \CL$ for any $e = O(1)$, which remains the frontier of our understanding.

In this work we completely characterize lossy catalytic space ($\LCSPACE[s,c,e]$) in terms of ordinary catalytic space ($\CSPACE[s,c]$). We show that
$$\LCSPACE[s,c,e] = \CSPACE[\Theta(s + e \log c), \Theta(c)]$$
In other words, allowing $e$ errors on a catalytic tape of length $c$ is equivalent, up to a constant stretch, to an equivalent errorless catalytic machine with an additional $e \log c$ bits of ordinary working memory.

As a consequence, we show that for any $e$, $\LCL[e] = \CL$ implies $\SPACE[e \log n] \subseteq \ZPP$, thus giving a barrier to any improvement beyond $\LCL[O(1)] = \CL$. We also show equivalent results for \textit{non-deterministic} and \textit{randomized catalytic space}.
\end{abstract}

\clearpage



\section{Introduction} \label{sec:intro}

\subsection{Catalytic computation}
Within space-bounded computation, the \textit{catalytic computing} framework,
first introduced by Buhrman, Cleve, Kouck{\'{y}}, Loff, and 
Speelman~\cite{BuhrmanCleveKouckyLoffSpeelman14},
models the question of whether or not full memory can be a computational
resource. Their main object of study is a \textit{catalytic logspace} ($\CL$)
machine, in which a traditional logspace-bounded Turing machine is given
access to a second work tape, polynomial in length, called the catalytic tape;
while this tape is exponentially longer than the logspace work tape, it
is already full with some string $\tau$ at the outset, and this string
$\tau$ must be preserved by the overall computation.

Surprisingly, \cite{BuhrmanCleveKouckyLoffSpeelman14} show that $\CL$
can be much more powerful than $\L$, with the catalytic tape
being at least as powerful a resource as non-determinism
($\NL \subseteq \CL$), randomness ($\BPL \subseteq \CL$), and more
($\TCo \subseteq \CL$). They also showed that its power is nevertheless
limited and falls far short $\PSPACE$, namely $\CL \subseteq \ZPP$.

This work spawned a long sequence of explorations of the
power of catalytic space. Given the base model of $\CL$
there are many possible variations and structural questions
to be asked, such as the power of
randomness~\cite{DattaGuptaJainSharmaTewari20,CookLiMertzPyne24},
non-determinism~\cite{BuhrmanKouckyLoffSpeelman18},
non-uniformity~\cite{Potechin17,RobereZuiddam21,CookMertz22,CookMertz24},
and other variants~\cite{GuptaJainSharmaTewari19,BisoyiDineshSarma22}.
There have also been many works connecting
the catalytic framework to broader questions in complexity theory, such as
space-bounded derandomization~\cite{Pyne23,DoronPyneTell24,LiPyneTell24},
as well as adaptations of catalytic techniques to solve
longstanding open questions such as
compositional upper bounds for space~\cite{CookMertz20,CookMertz21,CookMertz24}
(see \cite{Koucky16,Mertz23} for surveys on the topic).

\subsection{Lossy catalytic computation}

Besides these more standard structural questions, there are also
catalytic variants which are more specific to the catalytic space
restriction. In particular, Gupta et al.~\cite{GuptaJainSharmaTewari24}
initiated the study of \textit{lossy} catalytic computing, wherein
the catalytic tape need not be exactly reset to its initial configuration.
This model, which we refer to as $\LCSPACE$, essentially asks how
robust the core definition of catalytic space is to seemingly small relaxations;
for example, in the \textit{quantum} 
setting 
some computation error (albeit of a different form) is necessary for
converting between different definitions based on 
allowed operations. 

To begin, note that $\CL$ with $e \leq \poly(n)$ errors trivially contains 
the class $\SPACE[e]$ by simply erasing the first $e$ bits of the
catalytic tape and using them as free memory. Because
we have not managed to prove that any space-bounded class beyond
$\L$ which is contained in $\ZPP$, we should not expect
to be able to prove $\CL$ is the same as $\CL$ with $e = \omega(\log n)$ errors.
The question, then, is to understand where, in the range of
$e = 0$ to $e = O(\log n)$, is the acceptable number of errors
that $\CL$ can provably tolerate.

As an initial answer to the previous question, \cite{GuptaJainSharmaTewari24}
show that $\CL$ gains no additional power from allowing any constant
number of errors on the catalytic tape, i.e., $\LCL[O(1)] = \CL$.
This remains the frontier of our knowledge, and Mertz~\cite{Mertz23}
posed it as an open question to improve this result to any
superconstant number of errors, or, alternatively, to provide
evidence against being able to prove such a collapse.\footnote{We
cannot expect an unconditional separation between $\CL$ and any $\LCL$,
as even separating $\PSPACE$ from e.g.\ $\TCo (\subseteq \CL)$ remains wide open.}
Recently, Cook et al.~\cite{CookLiMertzPyne24} showed that a different
error-prone model, namely \textit{randomized} $\CL$, is no more powerful
than the base $\CL$ model, indicating that perhaps such an improvement
is possible.

\subsection{Our results}
In this work we completely characterize lossy catalytic space
in terms of ordinary catalytic space.
Let $\CSPACE[s,c]$ denote catalytic machines with free space $s$ and
catalytic space $c$, and let $\LCSPACE[s,c,e]$ be the same with
up to $e$ errors allowed in resetting the catalytic tape. We show
that these $e$ errors are equivalent to an additional $e \log c$
free bits of memory, up to constant factor losses.

\begin{theorem} \label{thm:main}
Let $s := s(n), c := c(n), e := e(n)$ be such that
$e \leq c^{1-\Omega(1)}$. Then
$$\LCSPACE[s,c,e] = \CSPACE[\Theta(s + e \log c),\Theta(c)]$$
\end{theorem}

Besides characterizing $\LCSPACE[s,c,e]$,
this allows us to understand the lay of the
land for $\LCL[e]$, i.e., $\CL$ with $e$ errors.
In particular, this recovers the result of \cite{GuptaJainSharmaTewari24},
which says that $\LCL[O(1)] = \CL$.
Furthermore, it gives intuition that this theorem is the best we can
hope for with respect to $e$, again assuming $\SPACE[e \log n]$ cannot be shown
to be in $\ZPP$ for any $e = \omega(1)$.

\begin{corollary} \label{cor:cl-for}
    For any $e := e(n)$,
    $$\LCL[e] = \CL \quad \mbox{implies} \quad \SPACE[O(e \log n)] \subseteq \ZPP$$
\end{corollary}




We also show that our proof extends to catalytic machines
with additional power beyond errors, namely \textit{non-deterministic}
and \textit{randomized} catalytic space.

\begin{theorem} \label{thm:main-resources}
    Let $\mathcal{C} \in \{\NCSPACE,\BPCSPACE\}$, and let
    $s := s(n), c := c(n), e := e(n)$ be such that
    $e \leq c^{1-\Omega(1)}$. Then
$$\LC[s,c,e] = \mathcal{C}[\Theta(s + e \log c),\Theta(c)]$$
\end{theorem}

We briefly remark that the $e \leq c^{1-\Omega(1)}$ restriction in
all our results is only needed to get the constant stretch in the
catalytic tape; we discuss the unrestricted setting in Section~\ref{sec:thm-rev}.

\subsection{Open problems}
\paragraph{Errors in expectation.}
A related question asked in \cite{Mertz23} is whether or not $\CL$
is equivalent to $\CL$ with $O(1)$ errors allowed \textit{in expectation}
over all starting catalytic tapes.
This represents a different notion of distance between catalytic
tapes, in opposition to Hamming distance, that may be more applicable
to settings such as quantum computation.
However, no results are known for expected errors, and all
techniques in this paper fail to restore the tape in pathological
cases where a few starting tapes end up with potentially 
many errors.

\paragraph{Randomized error-prone catalytic space.}
Recent work of Cook et al.~\cite{CookLiMertzPyne24} shows that
$\CSPACE[s,c] = \BPCSPACE[O(s),\poly(c)]$, which, in conjunction with
Theorem~\ref{thm:main-resources}, seems to indicate
that our theorems can be unified to show the connection between ordinary
$\CSPACE$ and $\CSPACE$ which is both randomized and lossy,
i.e.\ $\CSPACE[s + e \log c,c] = \LBPCSPACE[O(s),\poly(c),e]$.
This would characterize how deterministic catalytic
space handles both natural kinds of ``error'', namely both error in
the output from the randomness and error in resetting the catalytic tape.

However, the proof of \cite{CookLiMertzPyne24} only works when $c = 2^{\Theta(s)}$,
and our connection to error-prone space incurs an $e \log c$ blowup in
free space, putting us outside this regime. A generalization of
their result, i.e.\ showing $\CSPACE[s,c] = \BPCSPACE[O(s),\poly(c)]$
for \textit{every} $s$ and $c$,\footnote{Note that the polynomial blowup allowed
in the catalytic tape means this result would not yield novel derandomization
for ordinary space; even for $s, c = O(\log n)$ this would only show
that derandomization overheads can be pushed into a polylogarithmic length
catalytic tape, which was already shown by Pyne~\cite{Pyne23}.}
would tie off this connection.

\paragraph{Lossy catalytic branching programs.}
Due to the flexibility in the conditions of Theorem~\ref{thm:main},
the results of Theorem~\ref{thm:main-resources} are likely to extend
to other settings catalytic settings; for example, it is immediate to
extend both results to $\CSPACE$ with \textit{advice}. We focus on
non-determinism and randomness simply because these are two of the
most well-studied catalytic variants, and future works are free to
adapt these proofs to their own settings.

In terms of notable omissions, however,
one setting where one direction does not yet extend, and which is very related to
advice, is the \textit{catalytic branching program} model, which is
a syntactic, and by extension non-uniform, way of capturing $\CSPACE$.
The issue here is simply that such machines can read and write their
entire work tape in one step, which our simulation of
$\CSPACE$ by $\LCSPACE$ is unequipped to handle.
As we will note in Appendix~\ref{sec:reverse},
showing such branching programs are \textit{reversible} would be
sufficient to close this off.

\paragraph{Exact Simulation Space Requirements} In the current simulation of errors using clean space, we use $4e \log c$ clean space. By contrast, in our simulation of clean space using errors, we use only $e$ more errors. If errors can be simulated in clean space $e \log c$ instead, then there is only very low overhead in switching between the two perspectives. This would tighten the correspondence between errors and space that we establish. However, since the distance between two codewords required to correct $e$ errors is $2e + 1$, a different error correction code would be necessary to reach clean space $e \log c$.
\section{Preliminaries}

We begin by defining catalytic machines as introduced by
Buhrman et al.~\cite{BuhrmanCleveKouckyLoffSpeelman14}.

\begin{definition}[Catalytic space]
A \textit{catalytic Turing Machine} is a space-bounded Turing machine
with two work tapes: 1) a read-write work tape of length $s(n)$ which
is initialized to $0^{s(n)}$, and
2) a read-write \textit{catalytic tape} of length $c(n) \leq 2^{s(n)}$ which is
initialized to an arbitrary state $\tau \in \{0,1\}^{c(n)}$.
On any input $x \in \{0,1\}^n$ and initial catalytic state $\tau$, a catalytic
Turing machine has the property that at the end of the computation
on input $x$, the catalytic tape will be in the initial state $\tau$.
\end{definition}

In this work we focus on a relaxation of catalytic space by
Gupta, Jain, and Sharma~\cite{GuptaJainSharmaTewari24},
where we are allowed to make some errors in resetting the catalytic tape.

\begin{definition}[Lossy catalytic space]
A \textit{lossy catalytic Turing Machine with $e(n)$ errors} is a
catalytic machine where at the end of the computation on
any input $x \in \{0,1\}^n$ and initial catalytic state $\tau$, 
instead of requiring that the catalytic tape be in state $\tau$,
the catalytic tape can be in any state $\tau'$
such that $\tau$ and $\tau'$ differ in at most $e(n)$ locations.
\end{definition}

Lastly we specify the basic complexity classes arising from our two
catalytic definitions, as well as their specification to the
``logspace'' setting, where most research interest at the moment lies.

\begin{definition}
We write
\begin{itemize}
    \item $\CSPACE[s,c]$: the class of languages which can be recognized by catalytic
    Turing Machines with work space $s := s(n)$ and catalytic space $c := c(n)$.
    \item $\LCSPACE[s,c,e]$: the class of languages which can be recognized by
    lossy catalytic Turing Machines with work space $s := s(n)$,
    catalytic space $c := c(n)$, and $e := e(n)$ errors.
\end{itemize}
We additionally write
\begin{itemize}
    \item $\CL := \CSPACE[O(\log n), \poly n]$
    \item $\LCL[e] := \LCSPACE[O(\log n), \poly n, e]$
\end{itemize}
\end{definition}
\section{Main theorem}
\label{sec:thm}

In this section we will prove Theorem~\ref{thm:main}.
We will do so via a simulation argument for each direction in turn.


\subsection{Simulating errors with space}
\label{sec:thm-for}

First, we show that $\LCSPACE[s,c,e] \subseteq \CSPACE[O(s + e \log c),O(c)]$.
In fact, we will not need any increase in the length of our catalytic tape.

\begin{theorem}\label{thm:main-for}
    Let $s := s(n), c := c(n), e := e(n)$. Then
    $$\LCSPACE[s,c,e] \subseteq \CSPACE[s + O(e \log c),c]$$
\end{theorem}

We note that this was also proven in \cite{GuptaJainSharmaTewari24}
for the case of $\LCL[O(1)]$, but we will pursue a different
proof, based on error-correcting codes, which will allow us to
generalize to other catalytic models in Section~\ref{sec:corollaries}.

\begin{proof}
Let $M_e$ be an $\LCSPACE[s,c,e]$ machine.
We will devise a $\CSPACE[s + O(e \log c), c]$ machine $M_0$ which simulates $M_e$.
Note that in this section, we will not use our parameter restriction on $e$;
this direction holds for every setting of $s$, $c$, and $e$. We will presume that $e \leq \frac{c}{\log(c)}$, otherwise the inclusion becomes trivial.

Our simulation will go via an error-correcting code.
In particular we will use \textit{BCH codes}\footnote{Technically
because of our parameters, they can even be considered Reed-Solomon
codes, which are a special case of BCH codes; nevertheless we
follow the presentation of the more general code form.} ($\BCH$),
named after Bose, Ray-Chaudhuri, and Hocquenghem~\cite{bose1960class,hocquenghem1959codes},
which we define as per~\cite{DodisReyzinSmith04,DodisOstrovskyReyzinSmith06}.
We define the mapping $\BCH$ and prove the following lemma in
Appendix~\ref{sec:arithmetic} (see Corollary~\ref{cor:existance_BCH}, Lemma~\ref{lemma:encoding} and Lemma~\ref{lemma:decoding}).

\begin{lemma} \label{lem:bch}
    Let $q := 2^{\lceil \log(c + e)\rceil}$. There exists a mapping
    $\BCH: \mathbb{F}_q^q \rightarrow \mathbb{F}_q^q$ with the following operations:
    \begin{itemize}
    \item \textbf{Encoding:} $\Enc_{\BCH}$
    takes as input a string $S$ of length $c$, plus an additional
    $(2e + 1)\lceil \log(c + e) \rceil$ bits initialized in $0$,
    and outputs a codeword $S_{enc}$:
    \[
        S + [0]_{(2e + 1)\lceil \log(c + e) \rceil} \rightarrow_{\Enc}
        S_{enc}
    \]
    Furthermore, all outputs $S_{enc}$ generated this way have
    minimum distance $\delta := 2e + 1$ from one another.
    \item \textbf{Decoding:} $\Dec_{\BCH}$ takes as input a string
    $S_{enc}'$ of length $c +(2e + 1)\log(c + e)$, with the promise
    that there exists a string $S$ of length $c$ such that
    $\Enc_{\BCH}(S + [0]_{2e \log(c + e)})$ differs from $S_{enc}'$
    in at most $\delta/2 - 1 = e$ locations, and outputs this string $S$:
    \[
        S_{enc}' \rightarrow_{\Dec} S + [0]_{(2e+1)\log(c + e)}
    \]
    \end{itemize}
    Furthermore, both $\Enc_{\BCH}$ and $\Dec_{\BCH}$ can be
    computed in space $O(e \log c)$.
\end{lemma}
We now move on to the simulation of our $\LCSPACE[s,c,e]$ machine $M_e$.
Our $\CSPACE[s + O(e \log c),c]$ machine $M_0$ acts as follows:
\begin{enumerate}
    \item \textbf{Initialization:} use the function $\Enc_{\BCH}$ to encode the initial state
    $\tau$ of the catalytic tape into a codeword, using
    $(2e+ 1)\lceil\log(c + e)\rceil$ additional bits from clean space,
    \[
        \tau + [0]_{(2e + 1)\lceil\log(c + e)\rceil} \rightarrow_{\Enc} \tau_{enc}.
    \]
    
    \item \textbf{Simulation:} Run $M_e$ using clean space $s$ and the first $c$ bits of $\tau_{enc}$
    as the catalytic tape. When $M_e$ finishes the calculation, we record the answer
    in a bit of the free work tape. The catalytic tape is, at this point,
    in a state $\tau_{enc}'$ which differs in at most $e$ locations from $\tau_{enc}$.

    \item \textbf{Cleanup:} use the function $\Dec_{\BCH}$ to detect and correct
    our resulting catalytic tape $\tau_{enc}'$:
    \[
        \tau_{enc}' \rightarrow_{\Dec} \tau + [0]_{(2e+1)\lceil \log(c + e)\rceil}
    \]
    Once we finish this process, we output our saved answer and halt.
\end{enumerate}
The correctness of $M_0$ is clear, as it gives the same output as $M_e$.
By our error guarantee on $M_e$ and the correctness of $\Dec$,
our catalytic tape is successfully reset to $\tau$.
Our catalytic memory is $c$ as before, while for our free work space
we require $s$ bits to simulate $M_e$,
an additional $(2e + 1)\lceil \log(c+e)\rceil = (2+o(1))e \log c$
zero bits for our codewords, and $O(e \log c)$ space for
$\Enc_{\BCH}$ and $\Dec_{\BCH}$, for $s + O(e \log c)$ space in total.
\end{proof}

\begin{note}
    There is an alternative proof of this point, one which gets better parameters
    and relies on an interesting characterization of space, namely the
    \textit{reversibility} of space.
    This proof is a simplification and extension of the one originally provided
    in \cite{GuptaJainSharmaTewari24}, and we provide it
    in Appendix~\ref{sec:reverse} for those interested.
\end{note}

\subsection{Simulating space with errors}
\label{sec:thm-rev}

We now show the other direction of Theorem~\ref{thm:main},
i.e.\ $\CSPACE[s + e \log c,c] \subseteq \LCSPACE[O(s),O(c),O(e)]$.

\begin{theorem} \label{thm:main-rev}
Let $s := s(n), c := c(n), e := e(n)$, and $\epsilon > 0$ be such that
$e = o(c^{\epsilon/(1+\epsilon)})$. Then
$$\CSPACE[s + e \log c,c] \subseteq \LCSPACE[s + \log c,(1+o(1))c,(1+\epsilon)e]$$
\end{theorem}

Since $s \geq \log c$ by the definition of a catalytic machine,
this achieves the reverse direction of Theorem~\ref{thm:main} with very
small blowups in $s$ and $c$, and for $e$ bounded by a small polynomial
in $c$ we get a negligible error blowup as well. Note that we allow
$\epsilon > 1$, and so our proof is not limited to $e < c^{1/2}$; however,
we will pay for larger values of $e$ in the error blowup, and for
$e = c^{1-o(1)}$ this factor becomes superconstant.

To understand our construction, we will first prove a version with
looser space parameters. This result is incomparable to Theorem~\ref{thm:main-rev};
although we lose a factor of $e$ in our catalytic space, in exchange
we have no restrictions on $e$ and no loss in $e$ either.

\begin{theorem} \label{thm:main-rev-easy}
Let $s := s(n), c := c(n), e := e(n)$ be such that $c$ is a power of 2. Then
$$\CSPACE[s + e \log c,c] \subseteq \LCSPACE[s + (\log e + \log c + 2),(e+1)c,e]$$
\end{theorem}

\begin{proof}
Let $M_0$ be a $\CSPACE[s + e \log c, c]$ machine.
We will devise a $\LCSPACE[s + (\log e + \log c + 2),(e+1)c,e]$ machine
$M_e$ which simulates $M_0$. 
%
By the definition of a Turing machine, we will assume that in any time
step $M_0$ only reads and writes at most one bit on the work tape.

For any string $\sigma$ and set of indices $S$, let $\sigma^{\oplus S}$ denote
$\sigma$ after flipping the bits in the locations in $S$;
we abuse notation for singleton sets $S$ and use
$\sigma^{\oplus j}$ in place of $\sigma^{\oplus \{j\}}$.
Our key construction will use the following folklore\footnote{This
construction is based on the solution to the so-called
``almost impossible chessboard puzzle''; interested readers
can find the setup and solution in videos on the YouTube channels
3Blue1Brown (\url{https://www.youtube.com/watch?v=wTJI_WuZSwE})
and Stand-up Maths (\url{https://www.youtube.com/watch?v=as7Gkm7Y7h4}).
It can also be seen as the syndrome of the Hamming code.} construction:
\begin{lemma} \label{lem:chessboard}
    For every $k$, there exists a mapping $\mem: \{0,1\}^{2^k}
    \rightarrow \{0,1\}^k$, computable in space $k+1$, such that
    the following holds: for any $\tau \in \{0,1\}^{2^k}$ and
    any $S \subseteq [k]$,
    $$\mem(\tau)^{\oplus S} = \mem(\tau^{\oplus \val(S)})$$
    where $\val(S) = \sum_{i \in S} 2^i$ is the value in $[2^k]$
    given by the characteristic vector of $S$.
\end{lemma}
Intuitively, Lemma~\ref{lem:chessboard} gives us an easily computable
mapping where any transformation of the $k$-bit output string
can be achieved by flipping one bit of the $2^k$-bit input string,
with the location of this single bitflip being determined only
by the locations where the current and target output strings differ.
\begin{proof}[Proof of Lemma~\ref{lem:chessboard}]
    For each $J \in [2^k]$, consider the bitstring $\bin(J) \in \{0,1\}^k$
    corresponding to the binary representation of $J$. We will define
    our mapping $\mem$ for each input $\tau \in \{0,1\}^{2^k}$ as
    $$\mem(\tau)_j = \bigoplus_{J \in [2^k]} \bin(J)_j \cdot \tau_J$$
    This is clearly computable in space $k+1$, as we need only store $J$
    and our current sum. Now note that for any $S$, flipping the value
    $\tau_{\val(S)}$ flips every $\mem(\tau)_j$ value where $\bin(\val(S))_j = 1$,
    which are exactly the values $j \in S$, and leaves all other $\mem(\tau)_j$
    values unchanged.
\end{proof}

At a high level, our $\LCSPACE[s + (\log e + \log c + 2),(e+1)c,e]$
machine $M_e$ works as follows.
$M_e$ will use its $s$ bits of free memory and the first $c$ bits of
catalytic memory to represent their equivalent blocks in $M_0$,
i.e. $s$ bits of free memory and $c$ bits of catalytic memory.
We will break the remaining $e \cdot c$ bits of our catalytic tape of
$M_e$ into $e$ blocks $B_1 \ldots B_e$ of size $c$ each, and we
denote by $\tau_i$ the memory in block $B_i$.
We apply Lemma~\ref{lem:chessboard} with $k = \log c$---recall that
we assume $c$ is a power of 2---and
use each $\mem_i := \mem(\tau_i)$ to represent an
additional $\log c$ bits of free memory, giving us
an additional $e \log c$ bits of memory in total.
Our additional workplace memory will be used to compute the
mapping, serve as pointers, and other assorted tasks.

Before formally stating $M_e$, we mention a special case of the
construction of Lemma~\ref{lem:chessboard},
which will allow us to use it in a reversible operational manner.

\begin{claim} \label{claim:revert-chessboard}
    Let $\tau_0, \tau_1, \tau_2 \ldots \tau_{t-1}, \tau_t \in \{0,1\}^{2^k}$
    be such that 1) $\tau_i$ and $\tau_{i-1}$ differ in exactly one coordinate for all
    $i \in [t]$; 2) $\mem(\tau_i)$ and $\mem(\tau_{i-1})$ differ in exactly one coordinate
    for all $i \in [t]$; and 3) $\mem(\tau_t) = \mem(\tau_0)$. Then $\tau_t = \tau_0$.
\end{claim}
\begin{proof}
    For all $i \in [t]$, let $b_i \in [k]$ be the location where
    $\mem(\tau_i)$ and $\mem(\tau_{i-1})$ differ, and let $\beta_i = 2^{b_i}$
    be the location where $\tau_i$ and $\tau_{i-1}$ differ.
    Since $\mem(\tau_t) = \mem(\tau_0)$, each value $j \in [k]$
    must appear an even number of times in the list $b_1 \ldots b_t$,
    and since flipping any location $j \in [k]$ can only be obtained
    by flipping a unique location $J \in [2^k]$, it follows that
    each value $J \in [2^k]$ must appear an even number of times
    in the list $\beta_1 \ldots \beta_t$. This means that $\tau_t$ is $\tau_0$ with
    each bit flipped an even number of times, or in other words $\tau_t = \tau_0$.
\end{proof}

We now concretely define our machine $M_e$:
\begin{enumerate}
    \item \textbf{Initialization:} for each block $B_i$,
    calculate $\mem_i$ and flip the $\mem_i$th element of $B_i$:
    \[
        \tau_i \rightarrow \tau_i^{\oplus \mem(\tau_i)} \qquad \forall i \in [e]
    \]
    Define $\tau_i^{enc}$ to be the memory after this process. Note that we now have
    exactly $e$ errors on the tape, one in each $\tau_i^{enc}$, and we are guaranteed that
    \[
        \mem(\tau_i^{enc}) = 0^{\log c} \qquad \forall i \in [e]
    \]
    
    \item \textbf{Simulation:} run $M_0$ using $s$ free work bits and $c$ catalytic bits,
    with the concatenation of the values $\mem_i$ as the other $e \log c$ free work bits.
    To do this, whenever we read or write a bit in our $e \log c$ bits of memory,
    we find the $B_i$ responsible for this bit, calculate $\mem_i$,
    and update $\tau_i$ using one bitflip to reflect how $\mem_i$ changes according to
    the operation of $M_0$:
    \[
        \tau_i^{enc} \rightarrow (\tau_i^{enc})^{\oplus 2^j} \qquad \mbox{update occurs in bit $j$ of block $i$}
    \]
    If $\mem_i$ is unchanged, we make no updates to $\tau_i$.
    
    \item \textbf{Cleanup:} when we reach the end of $M_0$'s computation,
    record the answer on the free work tape and set each $\mem_i$ value to
    $0^{\log c}$ one bit at a time:
    \[
        \tau_i^{enc} \rightarrow (\tau_i^{enc})^{\oplus 2^j} \qquad \forall i \in [e], j: \bin(\mem(\tau_i^{enc}))_j = 1
    \]
    Once we finish this process, we output our saved answer and halt.
\end{enumerate}
The correctness of $M_e$ is clear, as we output the same value as $M_0$.
Our catalytic space usage is $c + ec$ by construction, while our free space usage
is $s$ to simulate $M_0$, one extra bit to save the output, and
any additional space required to handle
the simulation of the additional $e \log c$ work bits.
In particular, we need $\log e$ bits for a pointer into $B_i$ and $\log c + 1$ bits
for the computation of $\mem_i$ by Lemma~\ref{lem:chessboard},
for a total space usage of $s + \log e + \log c + 2$ as claimed.

We also claim that our lossy catalytic condition is satisfied.
By the property of $M_0$, we make no errors on the $c$ catalytic bits
used for the simulation of $M_0$'s catalytic space.
The initialization step introduces at most one error per $\tau_i$, thus
giving at most $e$ errors on to the catalytic tape.
After the initialization step, each other update to $\tau_i$ corresponds
to changing a single bit in $\mem_i$, with the final value being the same
as the value after initialization. Thus by Claim~\ref{claim:revert-chessboard}
we restore the catalytic tape to its position after the
initialization phase, and so our end state corresponds to our original
catalytic tape with at most $e$ errors, i.e. those induced by the
initialization phase, as required.
\end{proof}

We now return to Theorem~\ref{thm:main-rev}, which requires only
a small modification of the above proof, namely to break the the catalytic tape into more, smaller blocks, which reduces its required length, at the cost of a few extra errors. This modification works because the number of pure bits represented is logarithmic in the length of the block, and so making the blocks smaller barely affects the number of bits represented (for example, $c/2$ bits in a block still lets you represent $\log(c) - 1$ bits, so half the size only loses one bit per block). 

\begin{proof}[Proof of Theorem~\ref{thm:main-rev}]
Let $M_0$ be a $\CSPACE[s + e \log c, c]$ machine.
We will devise a $\LCSPACE[s+\log c,(1+o(1))c,(1+\epsilon)e]$ machine
$M_e$ which simulates $M_0$, where $\epsilon$ satisfies
$e = o(c^{\epsilon/(1+\epsilon)})$.

An easy manipulation gives us $c = \omega(e^{1+1/\epsilon})$,
and so there exists a function $\delta = \omega(1)$ such
that $c \geq (\delta e)^{1+1/\epsilon}$.
We will have the same approach as Theorem~\ref{thm:main-rev-easy},
but now we use $(1+\epsilon) e$ blocks of length
$2^{\lceil \log c/(\delta e) \rceil}$ each, i.e., using
Lemma~\ref{lem:chessboard} for $k = \lceil \log c/(\delta e) \rceil$.
Since we introduce one error per block, the number
of errors the machine makes is at most $(1+\epsilon)e$,
while our total catalytic tape has length
$$c + (1+\epsilon) e \cdot 2^{\lceil \log c/(\delta e) \rceil}
\leq c + (1+\epsilon) e \cdot \frac{2c}{\delta e} = c \cdot (1+o(1))$$
Lastly, we can use our additional catalytic memory to
simulate a work tape of length
$$(1+\epsilon)e \cdot \log 2^{\lceil \log (c/\delta e)\rceil}
\geq e \cdot (1+\epsilon) \log \frac{c}{\delta e}$$
and by manipulating our starting assumption we get that
\begin{align*}
c & \geq (\delta e)^{1+1/\epsilon} \\
(c/\delta e)^{\epsilon} & \geq \delta e \\
\epsilon \log (c/\delta e) & \geq \log \delta e \\
\log c - \log \delta e + \epsilon \log (c/\delta e) & \geq \log c \\
e \cdot (1+\epsilon) \log (c/\delta e) & \geq e\log c
\end{align*}
thus giving us a simulation of $e \log c$ work bits as claimed.
The correctness and lossy catalytic condition can then be
argued as above, and our free space usage is $s$
plus an additional $\log (c \cdot o(1)) + 2 \leq \log c$ bits,
for a total space usage of $s + \log c$ as claimed. 
\end{proof}

\section{Further consequences}
\label{sec:corollaries}

With this, we have concluded our main theorem and proof. We now move
to corollaries and extensions.

\subsection{Lossy catalytic logspace with superconstant errors}
\label{sec:cl-err}

As stated in the introduction, it immediately follows from Theorem~\ref{thm:main}
that proving $\LCL[e] = \CL$ is likely difficult, if not false, for
superconstant values of $e$.

\begin{proof}[Proof of Corollary~\ref{cor:cl-for}]
    This follows immediately from the fact that
    \begin{align*}
        \LCSPACE[O(\log n),\poly n,e] & = \CSPACE[O(\log n + e \log(\poly n)), \poly n] \\
        & = \CSPACE[O(e \log n), \poly n] \\
        & \supseteq \SPACE[O(e \log n)]
    \end{align*}
    combined with the fact that $\CL \subseteq \ZPP$ 
    by~\cite{BuhrmanCleveKouckyLoffSpeelman14}.
\end{proof}





\subsection{Lossy catalytic space with other resources}
\label{sec:resources}

As mentioned in Section~\ref{sec:intro}, there are many extensions of the
base catalytic model besides $\LCSPACE$, such as randomized, non-deterministic,
and non-uniform $\CSPACE$. So far, however, there has been little discussion
of classes where more than one such external resource has been utilized.
In this section, we will discuss two of the aforementioned models---randomized
and non-deterministic $\CSPACE$---in the presence of errors.

\begin{definition}
Let $f$ be a Boolean function on $n$ inputs.
\begin{itemize}
\item A \textit{non-deterministic catalytic Turing machine} for $f$
is a catalytic machine $M$ which, in addition to its usual input
$x$, has access to a read-once witness string $w$,
of length at most exponential in the total space
of $M$, based on $x$, which has the following properties:
\begin{itemize}
    \item if $f(x) = 1$, then there exists a witness $w$ such
    that $M(x,w) = 1$
    \item if $f(x) = 0$, then for every witness string $w$,
    $M(x,w) = 0$
\end{itemize}
Furthermore, for every witness string $w$, $M(x,w)$ restores
the catalytic tape to its original configuration.
\item A \textit{randomized catalytic Turing machine} for $f$
is a catalytic machine $M$ which, in addition to its usual input
$x$, has access to a read-once uniformly random string $r$,
of length at most exponential in the total space
of $M$, such that
$$\Pr_r[M(x,r) = f(x)] \geq 2/3 .$$
Furthermore, for every witness string $w$, $M(x,w)$ restores
the catalytic tape to its original configuration.
\end{itemize}
We write
\begin{itemize}
    \item $\NCSPACE[s,c]$: the class of languages which can be recognized
    by non-deterministic catalytic Turing Machines with work space
    $s := s(n)$ and catalytic space $c := c(n)$.
    \item $\BPCSPACE[s,c]$: the class of languages which can be recognized
    by randomized catalytic Turing Machines with work space
    $s := s(n)$ and catalytic space $c := c(n)$.
\end{itemize}
Furthermore, for $\mathcal{C} \in \{\NCSPACE,\BPCSPACE\}$,
we define $\LC[s,c,e]$ to be the class $\mathcal{C}[s,c]$ with the catalytic
resetting definition replaced by the $\LCSPACE$ resetting definition, i.e.,
where $e$ errors are allowed to remain on the catalytic tape at the end
of any computation.
\end{definition}

Note that we allow the errors to depend on the witness and randomness, respectively.
Without delving too deep into these models, however, it
is clear that our proof of Theorem~\ref{thm:main} carries
over to all the above definitions.

\begin{proof}[Proof sketch of Theorem~\ref{thm:main-resources}]
    As earlier, we need to show both directions. We will prove the same two equivalences
    as in Theorems~\ref{thm:main-for} and \ref{thm:main-rev}, namely
    \begin{enumerate}
        \item $\LC[s,c,e] \subseteq \mathcal{C}[s + O(e \log c), c]$
        \item $\mathcal{C}[s + e \log c,c] \subseteq \LC[s + \log c,(1+o(1))c,(1+\epsilon)e]$
    \end{enumerate}
    Both directions will follow immediately via the same simulation as before.
    For the forward direction, we simulate our $\LC$ machine by a $\mathcal{C}$ machine
    as usual, and then at the last step we correct the changes using our BCH code as
    before; this works just as before because the code allow us to correct any $e$ errors
    on the catalytic tape, regardless of how they came about.

    For the reverse direction, again our $\mathcal{C}$ machine will only
    read or write one bit per time step, and so we use the same $mem_i$ approach to
    simulating our additional $e \log c$ bits of free memory, which does
    not change based on the operation of the rest of the machine. As in the
    forward direction, simulating the actual workings of the $\mathcal{C}$ machine
    via the $\LC$ machine, given our method of simulating the $e \log c$ additional
    bits of memory, is straightforward, and our resetting step at the end again
    resets our extra catalytic tape regardless of the computation path the $\mathcal{C}$
    machine takes.
\end{proof}


\section*{Acknowledgments}

We acknowledge useful conversations with Swastik Kopparty and Geoffrey Mon about error correction codes and how to apply them.

\DeclareUrlCommand{\Doi}{\urlstyle{sf}}
\renewcommand{\path}[1]{\small\Doi{#1}}
\renewcommand{\url}[1]{\href{#1}{\small\Doi{#1}}}
\bibliographystyle{alphaurl}
\bibliography{bibliography}


\appendix
\section{Simulating errors with space via reversibility}
\label{sec:reverse}

In this section we give an alternate proof of simulating $\LCSPACE$
via $\CSPACE$, with sharper parameters than those in Theorem~\ref{thm:main-for}.

\begin{theorem}\label{thm:error-via-reversing}
    Let $s := s(n), c := c(n), e := e(n)$. Then
    $$\LCSPACE[s,c,e] \subseteq \CSPACE[s + (e+1) \log c,c] .$$
\end{theorem}

For this proof, we need to invoke a property of space-bounded machines
known as \textit{reversibility}, which we define now.

\begin{definition}
    A Turing machine $M$ is (strongly) \textit{reversible}
    if the following conditions hold:
    \begin{enumerate}
        \item For any input $x$, every node $v$ in its
        configuration graph $G_x$ has both in-degree and
        out-degree at most one. Let $for_x(v)$ indicate the
        unique node with a directed edge $(v,for_x(v))$,
        and let $back_x(v)$ indicate the unique node with a
        directed edge $(back_x(v),v)$.
        \item There exist machines $M_{\rightarrow}$ and
        $M_{\leftarrow}$ such that for every node $v$ in the
        configuration graph of $M$, $M_{\rightarrow}(x)$
        sends $v$ to $for_x(v)$ and $M_{\leftarrow}$ sends
        $v$ to $back_x(v)$.
    \end{enumerate}
\end{definition}

A classical result of Lange, McKenzie, and Tapp~\cite{LangeMckenzieTapp00}
shows that every $\SPACE[s]$ machine can be made reversible with no additional
space. Dulek~\cite{Dulek15} showed the same result for catalytic machines, while
Gupta et al.~\cite{GuptaJainSharmaTewari24} extended this to
catalytic machines with error; both of the latter results use a very
similar Eulerian tour argument to \cite{LangeMckenzieTapp00}.

\begin{lemma} \label{lem:rev}
    Let $M$ be a $\CSPACE[s,c]$ (resp.\ $\LCSPACE[s,c,e]$) machine recognizing language $L$.
    Then there exists a reversible $\CSPACE[s,c]$ (resp.\ $\LCSPACE[s,c,e]$) machine
    $M'$ which recognizes $L$.
\end{lemma}

In light of Lemma~\ref{lem:rev}, it seems that there is nothing interesting
to be said about $\LCSPACE$; after all, can we not simply reverse our machine
to the starting point, wherein there are no errors on the catalytic tape?
While this is technically true, there may be many different starting configurations
which reach the same halting state $(\tau,v)$. All such start states can and will
be reached by running $M_{\leftarrow}$ from $(\tau,v)$ for long enough,
but without knowledge of which particular start state
we began with, this na\"{i}ve reversing procedure cannot reset our
catalytic tape free of error.

Nevertheless, a small tweak on this idea,
using our additional $(e+1) \log c$ bits, immediately works.

\begin{proof}[Proof of Theorem~\ref{thm:error-via-reversing}]
Let $M_e$ be a  $\LCSPACE[s,c,e]$ machine, and by Lemma~\ref{lem:rev}
we will assume $M_e$ is reversible.
We will devise a $\CSPACE[s + (e+1) \log c, c]$ machine $M_0$ which simulates $M_e$.

We will assume without loss of generality that all start and end states of
a catalytic machine $M$ are distinguished; for example, we traditionally
assume any state with an all-zeroes work tape is a start state.
We write $\start(\tau)$ to
indicate the unique start state of $M$ with initial catalytic tape $\tau$,
while we write $\halt_x(\tau)$ to indicate the unique end state
reached by $M$ from initial state $\start(\tau)$ on input $x$.

Now let $S_{x,(\tau,v)} := \{\start(\tau_i)\}_{i}$ be the set of start states
such that $\halt_x(\tau_i) = (\tau,v)$.
Since $M_e$ is an $\LCSPACE[s+\log c,c,e]$
machine, each $\tau_i$ can differ from $\tau$ in at most $e$ locations, and thus
$$|S_{x,(\tau,v)}| \leq {c \choose \leq e} \leq \frac{c^{e+1}}{2}$$
Our machine $M_0$ thus works as follows:
\begin{enumerate}
    \item initialize a counter $num\_start$ with $\log {c \choose \leq e}$ bits to 0
    \item simulate $M_e$ using $s \log c$ work bits and $c$ catalytic bits,
    incrementing $num\_start$ each time we encounter a start state
    $\start(\tau_i)$, until we reach an end state $(\tau,v)$
    \item record our answer and run $M_e$ in reverse,
    decrementing $num\_start$ each time we encounter a start state
    \item halt when we reach a start state and $num\_start = 0$, and
    return our recorded answer
\end{enumerate}
Clearly our algorithm outputs the correct answer, resets the catalytic tape exactly, and uses at most $s + 1 + (e + 1)\log c - 1$ bits of work memory plus $c$ bits of catalytic memory.
\end{proof}

We also note that reversibility can be used for Theorem~\ref{thm:main-rev},
taking the place of the one-bit write condition.

\begin{proof}[Proof sketch of Theorem~\ref{thm:main-rev}]
We will construct an $\LCSPACE$ machine $M_e$ to simulate our
$\CSPACE$ machine $M_0$. By Lemma~\ref{lem:rev} we will assume
$M_0$ is reversible.
We will act as in the old proof of Theorem~\ref{thm:main-rev},
but instead of relying on Claim~\ref{claim:revert-chessboard},
i.e. resetting via flipping each bit an even number of times, we will
simply run $M_0$ backwards at the end of our simulation.
More formally our machine $M_e$ works as follows:
\begin{enumerate}
    \item initialization: for each block $B_i$, calculate $mem_i$ and flip the $mem_i$th element of $B_i$. Note that after this procedure, each $mem_i$ is the string $0^{\log c}$ and we have exactly $e$ errors on the tape.
    \item simulate $M_0$ using $s$ free work bits and $c$ catalytic bits, with the concatenation of the values $mem_i$ as the other $e \log c$ free work bits. To do this, whenever we read or write to memory in $B_i$, we calculate the $mem_i$ value containing this bit and update $mem_i$ according to the operation of $M_0$, using one bitflip to $\tau_i$ if $mem_i$ is changed and doing nothing otherwise. 
    \item when we reach the end of $M_0$'s computation, we record our answer on the free work tape and run $M_0$ in reverse, outputting our saved answer when we finish.
\end{enumerate}

By the construction in Lemma~\ref{lem:chessboard}, it is clear that
if we flip a bit $b$ in $\tau$ to transform $mem(\tau)$ into $mem(\tau)'$,
flipping the same bit $b$ causes us to transform back into $mem(\tau)$.
Therefore, the reversing procedure exactly resets the catalytic tape
used to simulate our $e \log c$ bits of memory as before.
\end{proof}

We defer these discussions to the appendix for two reasons.
First, the error-correcting approach more directly
applies in both directions of Theorem~\ref{thm:main};
while Lemma~\ref{lem:rev} connects to Theorem~\ref{thm:main-rev},
the Hamming code connection via Lemma~\ref{lem:chessboard} is
still the driving force behind the construction.
Second, the reliance on reversibility makes the proof unsuitable
to our later generalizations from Section~\ref{sec:resources};
in particular, both randomized and non-deterministic catalytic
computations are only reversible in a limited sense, one which
rules out using Lemma~\ref{lem:rev}. However, as discussed in
Section~\ref{sec:corollaries}, if some other model,
such as the catalytic branching program model, is amenable
to reversing, these proofs may provide a direct way to extend Theorem~\ref{thm:main}.
\section{Space Efficient Linear Algebra on Finite Fields}
\label{sec:arithmetic}

\subsection{The Space Complexity of Solving Linear Systems}

We prove the space efficiency of various common arithmetic and linear algebra operations necessary in order to encode and decode BCH codes. First, we introduce the concept of well-endowed rings \cite{BorodinCookPippenger83}. This allows us to use earlier results to argue about the efficiency of various operations on rings without having to reprove those ourselves. The fields of interests are fields of the form $GF(p^{r_n})$ for a fixed prime $p$ and a sequence $r_n$. Our results will apply to a field whose size increases asymptotically. Hence the uniformity of the calculations involved is important. But we assume that $p$ is fixed for all fields.

All these results are expressed in their asymptotic complexity in terms of the size of the ring or a length function, which may be seen as a measure of the number of bits necessary to write down a value in a ring.

\begin{definition}
    A length function $\rho$ for a ring $R$ is a function satisfying that for any $x, y \in R$
    \begin{enumerate}
        \item $\rho(x + y) \leq \max \{\rho(x), \rho(y)\} + O(1)$
        \item $\rho(xy) \leq \rho(x) + \rho(y) + O(\log \max \{\rho(x), \rho(y)\})$
    \end{enumerate}
    An example is the number of bits of an integer.
\end{definition}

From here we can define well-endowed rings as those with efficient implementations of addition, negation and multiplication.

\begin{definition}
    A ring $R$ with length function $\rho$ is \textit{well endowed} if there is a succinct uniform representation in which it has efficient implementations of addition, negation and multiplication. Addition and negation are considered efficient if they can be implemented in (logspace uniform) $\mathsf{NC}^0$ and multiplication is considered efficient if it can be implemented in (logspace uniform) $\mathsf{NC}^1$. The parameter for $\mathsf{NC}^1$ functions is always the length function of the ring.
\end{definition}

We now argue that basic arithmetic can be done space efficiently. This is done in the following steps. First, we argue that the polynomial ring $GF(\xi)$ is well endowed and therefore we can perform polynomial addition, negation and multiplication efficiently. Then we argue that we can use this to compute the remainder of polynomial division efficiently. This allows to find an irreducible polynomial to represent the field $GF(p^{r_n})$ in order to perform addition, negation and multiplication efficiently. We finally show that we can evaluate multiplicative inverses inefficiently and use this to do division. With inefficiently we mean in space $O(\log |F|)$ for a field $F$ whereas addition, negation and multiplication can be performed in space $O(\log \log |F|)$. 

\begin{lemma}
    \label{lemma:well_endowed_polys}
    For fixed $p$, the ring $GF(p)[\xi]$ is well endowed.
\end{lemma}
\begin{proof}
    We argue that finite fields are well-endowed rings. First observe that $GF(p)$ for a fixed~$p$ is always well endowed since the size of the ring is independent of $n$ so addition, negation and multiplication can be performed in $\mathsf{NC}^0$. By Proposition 3.9 from \cite{BorodinCookPippenger83} this means that polynomials over $GF(p)[\xi]$ are also well endowed. The length function here is $O(d 
)$ for a polynomial of degree $d$. Since they are well endowed, one can perform addition, negation and multiplication in space $O(\log d)$ for polynomials.
\end{proof}

We use this to compute the remainder.
\begin{lemma}
    \label{lemma:remainder}
    Given polynomials $N(\xi)$ and $D(\xi)$ in $GF(2)[\xi]$ of degree at most $r_n$, it is possible to compute the remainder $R(\xi)$ using an additional $r_n  + O(\log r_n)$ space. If we can overwrite $N(\xi)$ in place, the additional space necessary is $\lceil \log r_n \rceil + O(1)$. 
\end{lemma}
\begin{proof}
    Suppose that $\chi \in GF(p)$ is the leading coefficient of $D(x)$, we can compute and store $\chi^{-1}$ in constant space since $2$ is constant. We perform a kind of Gaussian elimination to compute the remainder:
    \begin{enumerate}
        \item If the degree of $D(\xi)$ exceeds the degree of $N(\xi)$ then return $N(\xi)$.
        \item Let $\psi$ be the leading coefficient of $N(\xi)$. Compute $N(\xi) + -\psi \chi^{-1} D(\xi)$ overwriting $N(\xi)$ in the process. 
            Since we use fixed field $GF(2)$, this can be done in constant depth.
            Repeating this for each coefficient uses $O(\log r_n)$ space for a counter. We use $O(1)$ to store $\psi$ and the coefficient of $N(\xi)$ during the computation. We then compute $N(\xi) + - \psi^{-1} D(\xi)$ coefficient by coefficient. 
        \item return to step 1.
    \end{enumerate}
    Overall, we manage to compute the remainder in space $r_n + O(1)$ by copying the final remainder to a new part of the space and then updating it in place during every iteration. If we can overwrite $N(\xi)$ in the process, then the additional space required is only $r_n$ to keep track of a counter. 
\end{proof}

We can now search for irreducible polynomials.
\begin{lemma}
    \label{lemma:irred_polys}
    Given a sequence of positive integers $r_n$ and a constant prime $p$, it is possible to find a degree $r_n$ irreducible polynomial in $GF(p)[\xi]$ in space $3 r_n + O(\log r_n)$. 
\end{lemma}
\begin{proof}
    It costs $d $ space to store a polynomial over $GF(p)$ of degree at most $d$. Therefore, one can iterate over all such polynomials. If we store two such polynomials and iterate over all pairs, the first can be a candidate irreducible polynomial, while the second can be a candidate factor of the first polynomial. By using Lemma \ref{lemma:remainder} to test whether or not the candidate irreducible polynomial is divisible by the candidate factor in additional space $r_n + O(\log r_n)$, we can test if the candidate irreducible polynomial is irreducible. This uses an additional $d $ space to store the remainder as it is calculated. Irreducible polynomials are guaranteed to exist, so we must find one eventually.
\end{proof}

Together these results allow us to do division in $GF(q^r)$.
\begin{lemma}
    \label{lemma:mult_inverse}
    Given a sequence of finite fields $F_n = GF(p^{r_n})$ for a constant prime $p$, it is possible to compute the multiplicative inverse of an element $x \in GF(p^{r_n})$ in additional space $4 r_n + O(\log r_n)$ counting the space needed to store the irreducible polynomial.
\end{lemma}
\begin{proof}
    If $x = 0$ then there is no multiplicative inverse.
    Otherwise, try multiplying $x$ by every other possible $y$ and taking the remainder using Lemma \ref{lemma:remainder} in place in space $4 r_n )$ until one finds a $y$ such that $xy = 1$. It takes $r_n $ space to iterate over all possible $y$. For every $x, y$, we use another register to store $xy$. Storing $xy$ needs an additional $2 r_n $ space, since we first need to compute the product as a product of polynomials and only take the remainder later. Computing $xy$ uses an additional space $O(\log r_n)$ since the ring of polynomials over $GF(p)$ is well endowed. Finally, we can use Lemma \ref{lemma:remainder} to take the remainder in place in only $\lceil \log r_n \rceil + O(1)$. 
\end{proof}

We can now finally solve linear systems.
\begin{lemma}
    \label{lemma:linsys}
    Given a sequence of finite fields $F_n = GF(p^{r_n})$ for a constant prime $p$, it is possible to solve a linear system of $t_n$ equations and $t_n$ unknowns in $2t_n r_n + 5 r_n + O(\log^2 r_n + \log t_n)$ space if $t_n = O(|F_n|)$ and we count the space used to store the irreducible polynomial for our representation of $GF(p^{r_n})$.
\end{lemma}
\begin{proof}
    By Proposition 4.2 from \cite{BorodinCookPippenger83} it is possible to compute the determinant over a well endowed ring in $\mathsf{NC}^2$. By Lemma \ref{lemma:well_endowed_polys}, we can perform this computation by treating elements of $F_n$ as polynomials over $GF(p)$ first. By Theorem 4 from \cite{Borodin77} this can be done in space $O(\log^2 \log |F_n|)$. 
    The cost of storing the polynomial representing the remainder is $t_n r_n $ since each entry of $F_n$ uses $r_n $ bits and the determinant is a sum of the product of at most $t_n$ elements.
    Then we can find an irreducible polynomial (or preferably access one that has been precomputed) in space $3r_n + O(\log r_n)$ by Lemma \ref{lemma:irred_polys} and take the remainder in place with using additional space $\lceil \log t_n r_n \rceil$ by Lemma \ref{lemma:remainder}. This allows us to calculate determinants of $t_n \times t_n$ matrices over $F_n$.

    We can then use Cramer's rule to solve our equation. Cramer's rule gives the solution to our system as a fraction of determinants. We use Lemma \ref{lemma:mult_inverse} to compute the multiplicative inverse of the denominator in additional space $3 r_n + O(\log r_n)$. Multiplying the inverse of the denominator with the numerator using their properties as well-endowed rings and writing into a (double) register used for the multiplicative inverse in additional space $O(\log r_n)$. If we then take the remainder in place in space $O(\log r_n)$, we can compute the solution to the system in $F_n$.
\end{proof}

\subsection{An Overview of BCH Codes}

The codes used to correct errors in our catalytic tape are so-called BCH codes (Bose–Chaudhuri–Hocquenghem codes) as described by \cite{DodisOstrovskyReyzinSmith06}. A BCH code has the following components.

\begin{enumerate}
\item An alphabet represented by a `small' field $GF(q)$.
\item Codeword length $n = q^m - 1$. Each position of the codeword is represented by a member of $F^*$ where $F^*$ is the multiplicative group of $F = GF(q^m)$ for a fixed value $m$. We may call $F = GF(q^m)$ the larger field.
\item Distance $\delta$.
\end{enumerate}

And we make the following choices.

\begin{enumerate}
\item We set $q = p^{r_n}$ for a prime number $p$ and $r_n$ that depends on the size of the input tape of the machine. Here $p$ is fixed and we set it to $p = 2$ in this work.
\item $m = 1$, therefore the small field equals the large field $F = GF(q)$.
\item $\delta = 2e + 1$. It is well known that one needs a distance of at least $2e + 1$ to be able to correct $e$ errors.
\end{enumerate}

Together these choices form a $[p^{r_n} - 1, p^{r_n} - 1 - \delta, \delta]$-code over an alphabet of size $p^{r_n}$. Ensuring that $p = 2$ means that we can interpret the catalytic tape as a sequence of elements in $GF(q) = GF(2^{r_n})$. Furthermore, we wish to have the property that by extending a word by a small amount we can turn any word into a codeword. We observe that codewords are defined as words that satisfy the following property for $i = 1, \dots, \delta - 1$.

\begin{equation}
    s_i = \sum_{x \in F^*} d_x x^i = 0
\end{equation}

Here the $d_x$ represents the value of the codeword stored at position $x$. Now presume we have a word of length $n$ then we extend the word by adding entries, we call the list of added entries $C \subseteq F^*$. The added values can be set arbitrarily, therefore we obtain the following equations:

\begin{equation}
    \label{eq:encoding_rough}
    s_i = \sum_{x \in F^*} d_x x^i 
    = \sum_{x \in F^* \setminus C} d_x x^i + \sum_{x \in C} d_x x^i
    = s_i' + \sum_{x \in C} d_x x^i 
    = 0
\end{equation}

for

\begin{equation}
\label{eq:si'}
    - s_i' = \sum_{x \in \mathcal{F}^* \setminus C} d_x x^i \,.
\end{equation}

We observe that for every value of $i = 1, \dots, \delta - 1$, we obtain an equation. Each equation is linear in the $d_x$ for $x \in C$. These are the new data points we must calculate in order to turn an arbitrary word into a codeword. Overall this yields the encoding linear system with parameter $\delta - 1$ and $i = 1, \dots, \delta - 1$
\begin{equation}
\label{eq:encoding}
\sum_{x \in C} d_x x^i
=
- s_i'
\end{equation}

In order to argue that a solution to this system always exists, we need the small field to equal the large field of the BCH code. This means $m = 1$. This is necessary because the value $s_i'$ lies in the large field of the code while the values of $d_x$ lie in the small field of the code. If these are the same, we can treat this as a linear algebra problem.

\begin{lemma}
\label{lemma:bch_codes}
By setting $|C| = 2e = \delta - 1$, adding this many members of the alphabet of a BCH code with $m = 1$, it is always possible to turn any string into a codeword. 
\end{lemma}
\begin{proof}
    As discussed, it is sufficient to show that Equation \ref{eq:encoding} always has a solution. In order to see this, observe that Equation \ref{eq:encoding} forms a linear system over the field $GF(q)$ and that the matrix of this system is a Vandermonde matrix. Vandermonde matrices are always invertible. Thus a solution to this system always exists. 
\end{proof}
\begin{corollary}
\label{cor:existance_BCH}
Let $S$ be a data string of $n$ bits and $e \leq \frac{1}{2} c /\log(c)$, then there exists a BCH code, with distance $\delta = 2e +1$ and codeword length $n + (2e + 1)\lceil\log(n + e) \rceil$.
\end{corollary}

\begin{proof}
We set $r_n = \lceil \log(n + e) \rceil$ therefore $q = 2^{\lceil \log(n + e) \rceil}$, therefore the alphabet of is of size $ 2^{\lceil \log(n + e) \rceil}$. We break the initial catalytic tape into blocks of length $\lceil \log(n + e) \rceil$, these blocks form the initial letters of the word. If $\lceil \log(n + e) \rceil \nmid n$, we pad the last block of the catalytic tape with additional $0$'s of free space, this requires at most $\lceil \log(n + e) \rceil - 1$ bits of free space. This gives a word consisting of $\lceil \frac{n}{(\lceil\log(n + e)\rceil)}\rceil$ letters. Now we use Lemma~\ref{lemma:bch_codes} and add $2e$ letters of size $\lceil \log(n + e) \rceil$, using $2e \lceil \log(n + e) \rceil$ of free space, such that these new members abide by Equation~\ref{eq:encoding}. This creates a codeword of length $n + (2e + 1)\lceil\log(n + e) \rceil$ as required. 
\end{proof}
\begin{remark}
Note that $q$ in this lemma is taken larger than strictly necessary. There are two requirements on $q$, namely $q$ is the total number of letters we can use to construct a code-word, and that $\log(q)$ the number of bits of which one codeword exists. In this work we set $r_n$ to the number of bits required to represent one letter of the word, which gives the following equation for $r_n$:
\begin{align*}
2^{r_n} \geq 2e + \frac{n}{r_n},
\end{align*}
which our choice of $r_n$ satisfies, but might not always be optimal.
\end{remark}
Given that this code exists and has the correct space complexity we will show that it can be space efficiently computed. Even before doing encoding and decoding, it is required to do an initialization step:
\begin{algorithm}
\caption{Initialization}
\label{alg:initialization}
\begin{algorithmic}[1]
    \State \textbf{Input}: $r \in \mathbb{N}$
    \State Compute an irreducible polynomial of degree $2$ in $GF(2)[\xi]$ via the procedure described in the proof of \ref{lemma:irred_polys}
    \State Pick and save an element that is not 0 and not 1 in $GF(2^{r})$. We can always pick this to be the polynomial $\xi$. \\
    \Return An irreducible polynomial of degree $2^r$ and a generator of the multiplicative group of $GF(2^{r})$.
\end{algorithmic}
\end{algorithm}

We now argue the initialization can be done space efficiently.

\begin{lemma}
    \label{lemma:initialization}
Given a sequence of fields $F_n = GF(2{r_n})$, Algorithm \ref{alg:initialization} can be performed in space $3r_n + O(\log r_n)$.
\end{lemma}
\begin{proof}
    We review each step of Algorithm \ref{alg:initialization} and review their space cost:
    \begin{enumerate}
        \item For step 1, use \ref{lemma:irred_polys} to find an irreducible polynomial in space $3 r_n + O(\log r_n)$. Only $r_n $ is needed to store the result.
        \item For step 2, we can pick and save the element of $GF(2)[\xi]$ corresponding to $\xi$. This uses $O(1)$ space if always done the same. This does not work for $r_n = 1$, but we can assume always $r_n > 2$.
    \end{enumerate}
\end{proof}
Encoding requires solving the linear equations given by Equation~\ref{eq:encoding}, finding the values $d_x$ for $x \in C$, the additional blocks that were appended. Solving these linear equations requires first calculating  the quantities $s_i'$, given by Equation~\ref{eq:si'}. We use the following algorithm to calculate a specific value $s_i$. By stopping prematurely, we can compute $s_i'$.

\begin{algorithm}
\caption{ComputeChecks}
\label{alg:compute_checks}
\begin{algorithmic}[1]
\State \textbf{Input}: Integer $i$ for $0 < i < \delta$, $0 \leq t < q$
\State Open five registers to store elements of $GF(2^{r_n})$ labelled $G, I, P, M, S, E$ for Generator, Index, Power, Multiplication, Sum and End. 
\State Set all registers except $G$ to 0.
\State Open two registers to store elements of $\{0, 1, \dots, \delta - 1\}$ called $C$ and one to store $i$ that never changes.
\State Assume $G$ stores a generator of the multiplicative group of $GF(p^r)$.
\State $E \leftarrow G^t$ via iterated multiplication. Use register $M$ as a counter.
\State $I \leftarrow G$
\State $P \leftarrow I^i$ via iterated multiplication. Use register $C$ as a counter in this process.
\State  $M \leftarrow P * d_x$.
\State $S \leftarrow S + M$
\State $P, M \leftarrow 0$.
\State $I \leftarrow I * G$
\State Return to step 8 until $E = G$.\\
\Return The value $s_i$ in register $S$ computed on the word.
\end{algorithmic}
\end{algorithm}
Now we give the space complexity of Algorithm \ref{alg:compute_checks}.
\begin{lemma}
    \label{lemma:compute_checks}
    Given a sequence of positive integers $r_n$, the $s_i$ and $s_i'$ can be computed in space $6 r_n + 2 \lceil \log \delta \rceil + O(\log r_n)$, counting the $r_n$ space used to store the generator of the multiplicative group of $GF(2^{r_n})$ and an irreducible polynomial to represent $GF(2^{r_n})$. We assume that the generator is simple so does not use much space.
\end{lemma}
\begin{proof}
    Every element of $GF(2^{r_n})$ uses space $r_n$. We use six of these. We also use two registers of size $\lceil \log 
    \delta \rceil$. Multiplication and addition use overhead $O(\log r)$. 
\end{proof}

We present the BCH encoding algorithm, and argue that it is space efficient.

\begin{algorithm}
\caption{$Encode_{BCH}$}
\label{alg:encoding}
\begin{algorithmic}[1]
    \State \textbf{Initialization}: We assume that Algorithm \ref{alg:initialization} has been performed in advance.
    \State Compute and store the $s_i'$.
    \State Solve Equation \ref{eq:encoding}.
    \State Store the solution. \\
    \Return (with the added entries, we have a codeword now)
\end{algorithmic}
\end{algorithm}

\begin{lemma}
    \label{lemma:encoding}
    It is possible to encode a word of length $2^{r_n} - \delta$ with an alphabet $F_n = GF(2^{r_n})$ and distance~$\delta$ as a codeword of length $2^{r_n}$ with an additional space overhead of $O(e \log n) = (4e + 6) r_n + O(\log^2 r_n)$, for $r_n = \log(e + c) \leq \log(c) + 1$. This implements the function $Encode_{BCH}$. 
\end{lemma}
\begin{proof}
    We look at the space complexity of every step of the encoding procedure. 
    \begin{enumerate}
        \item Initialization costs $3r_n  + O(\log r_n)$ space by Lemma \ref{lemma:initialization}.
        \item For step 2, use Algorithm \ref{alg:compute_checks}. This means we store $2e r_n $ values and use $6r_n  + 2\lceil\log(\delta)\rceil + O(\log r_n)$ space. 
        \item For step 3, we use Lemma \ref{lemma:linsys} which uses $(4e + 5) r_n + O(\log^2 r_n)$ space.
    \end{enumerate}
    Overall, this adds up to a space cost $(6e + 11) r_n +\lceil \log(\delta) \rceil + O(\log^2 r_n)$.
\end{proof}

That completes encoding. We now move to our analysis of decoding. We review the mathematics of the decoding.

We now describe the theory of the decoding algorithm. Decoding follows the procedure described in \cite{DodisReyzinSmith04,DodisOstrovskyReyzinSmith06} with some simplifications since we prioritize space over time.
First, the syndrome $\mathsf{syn}(p)$ of a message $p$ is computed. The syndrome is defined as the collection of the $s_i$ values defined before. From the syndrome we compute the support of the error, $\mathsf{supp}(p) = \{(x, p_x)_{x : p_x \neq 0}\}$ which is defined as the value of the error $p_x$ together with its position $x$. Then the error can be `subtracted' from the word to give back the original codeword. The error correction method only works if the number of errors is at most $(\delta - 1)/2$ and hence we set $\delta = 2e + 1$. It is important for space efficiency that we store only the support of the error, instead of a full error string which would require too much space. 
The support on the other hand uses exactly $O(e \log c)$ space.

The decoding algorithm is a variation of Berlekamp's BCH decoding algorithm. First, define the following polynomials using $M = \{x \in \mathcal{F}^* | p_x \neq 0\}$
\begin{align}
\sigma(z) &= \prod_{x \in M} (1 - xz)
          &
\omega(z) &= \sigma(z) \sum_{x \in M} \frac{p_x xz}{1 - xz}
\end{align}

which both have degree at most $|M| \leq (\delta - 1)/2$. Here $\sigma(z)$ is known as the error locator polynomial since the multiplicative inverses of its roots are the locations of the errors. Similarly, $\omega(z)$ is known as the evaluator polynomial since it gives the error since $\omega(x^{-1}) = p_x \prod_{y \in M, y \neq x} (1 - yx^{-1})$. Note that since these polynomials have no common zeroes, $\mathrm{gcd}(\sigma(z), \omega(z)) = 1$. 

It turns out that $\sigma(z)$ and $\omega(z)$ are the almost unique solutions to the congruence (with parameter $\delta - 1$)
\begin{equation}
    \label{eq:decode_linsys}
    S(z) \sigma(z) \equiv
    \omega(z) \pmod {z^\delta}
\end{equation}
where $S(z) = \sum_{l=1}^{\delta - 1} r_l z^l$. Suppose that $\sigma'(z),\omega'(z)$ are other solutions to this congruence then
\begin{equation}
    \omega(z) \sigma'(z) \equiv
    \sigma(z) S(z) \sigma'(z) \equiv
    \sigma(z) \omega'(z) \pmod {z^\delta} \,.
\end{equation}
Therefore if we restrict the degree of both $\omega(z)$ and $\sigma(z)$ to be polynomials of degree at most $(\delta - 1) / 2$ then as polynomials it is also true that $\omega(z) \sigma'(z) = \sigma(z) \omega'(z)$ and therefore that $\omega(z) / \sigma(z) = \omega'(z) / \sigma'(z)$. So if we also require that $\omega(z)$ and $\sigma(z)$ are relatively prime and $\sigma(z)$ has constant coefficient 1, then $\omega(z), \sigma(z)$ are unique. We call the linear system over $GF(q)$ from Equation \ref{eq:decode_linsys} the decoding linear system with parameter $\delta$.

After setting the constant term of $\sigma(z)$ to be 1, the above congruence gives a linear system with $\delta$ unknowns and $\delta$ equations with coefficients in the field $GF(q^m)$. We use almost the same procedure as described in the encoding step and making use of Lemma \ref{lemma:linsys} in order to solve this system. If less than $(\delta - 1) / 2$ errors are made, a solution is guaranteed to exist. However, we cannot force $\sigma(z)$ and $\omega(z)$ to be coprime in the linear system and as a result the solution may not be unique. Suppose $\sigma(z)$ and $\omega(z)$ have a common factor $\tau(z)$. Since $\omega(z)$ must have a constant coefficient 1, the constant coefficient of $\tau(z)$ must also be 1. But then $\tau(z)$ must have degree at least 1 and therefore $\sigma(z)$ and $\omega(z)$ both have degree at least 1 too high.
Therefore, if more than one solution to the system exists, some of these solutions will have too high a degree. But we can test whether or not the system has more than one solution by evaluating the determinant of the matrix of that system.
If the determinant is 0, we can repeat our procedure but now with $\delta$ replaced by $\delta - 2$ to get a smaller linear system. Repeating this procedure, either we find that the determinant is always 0, meaning that there are no errors to correct, or eventually that the determinant is nonzero and we can solve for polynomials $\sigma(z), \omega(z)$. 

Having solved for polynomials $\sigma(z)$ and $\omega(z)$ we can iterate over all possible values of $z \in \mathcal{F}^*$ to find all roots to $\sigma(z)$ and then compute their inverses using a similar procedure to that described in the encoding step. The evaluation of this polynomial can be done space efficiently, similar to the evaluation of $s_i$ but much simpler in fact. Afterwards, we can evaluate $\omega(z)$ to compute the errors. This is not necessary when $q = 2$ and the error is guaranteed to be 1. Once these have been computed, storing the support of the error is space efficient and the catalytic tape can be corrected. This completes the decoding step. This procedure is performed by the following algorithm, and we give it space complexity.

\begin{algorithm}
\caption{$Decode_{BCH}$}
\label{alg:decode}
\begin{algorithmic}[1]
    \State \textbf{Initialization}: We assume that Algorithm \ref{alg:initialization} has been performed in advance.
    \State Compute the syndrome using Algorithm \ref{alg:compute_checks}
    \State Compute the determinant $\Delta$ of linear system \ref{eq:decode_linsys} with $j = \delta - 1$. Use the method described in the proof of Lemma \ref{lemma:linsys}.
    \While{$\Delta = 0$ and $j > 0$}
    \State $j \leftarrow j - 2$
    \State Compute the determinant $\Delta$ of linear system \ref{eq:decode_linsys} with parameter $j$.
    \EndWhile
    \If{$\Delta = 0$}
    \State Terminate the algorithm (no errors detected).
    \EndIf
    \State Use Lemma \ref{lemma:linsys} to solve linear system \ref{eq:decode_linsys} with parameter $j$. 
    \For{$i \leftarrow 0; i < j; i \leftarrow i + 1$}
        \State Find the $i$th root, $x_i^{-1}$ of the error locator polynomial according to some ordering.
        \State Compute the quantity $\alpha_{x_i}^{-1} = \left( \prod_{y \in M, y \neq x} (1 - yx^{-1}) \right)^{-1}$
        \State Evaluate the evaluator polynomials and compute the errors by multiplying $\omega(x^{-1})$ by $\alpha_{x_i}^{-1}$.
        \State Correct the corresponding error.
    \EndFor \\
    \Return (up to $e$ errors have now been corrected)
\end{algorithmic}
\end{algorithm}

\begin{lemma}
    \label{lemma:decoding}
    Algorithm~\ref{alg:decode} can be performed with space overhead $O(e \log n) = (6 + 4e) r_n + O(\log^2 r_n)$, for $r_n = \log(e + c) \leq \log(c) + 1$, including the cost of the initialization using Algorithm \ref{alg:initialization}.
\end{lemma}
\begin{proof}
    We review the cost of Algorithm~\ref{alg:decode} step-by-step. Steps that use a trivial amount of space are omitted.
    \begin{enumerate}
        \item Initialization costs $3r_n + O(\log r_n)$ space by Lemma \ref{lemma:initialization}.
        \item By Lemma \ref{lemma:compute_checks}, the cost of Algorithm \ref{alg:compute_checks} is $5r_n + 2 \lceil \log \delta \rceil + O(\log s_i)$ not counting the space needed to store the irreducible polynomial. 
        \item Storing and computing a determinant using the method in the proof of Lemma \ref{lemma:linsys} costs $2e r_n + O(\log^2 r_n + \log e)$ space using access to an irreducible polynomial given in the initialization. The counter uses space $O(\log e)$.
        \stepcounter{enumi}
        \stepcounter{enumi}
        \item Reuse space from step 3.
        \stepcounter{enumi}
        \stepcounter{enumi}
        \item Solving the linear system by Lemma \ref{lemma:linsys} uses space $4e r_n + 5r_n + O(\log^2 r_n + \log e)$. We can reuse space used in step 3.
        \stepcounter{enumi}
        \item Evaluating a degree $2e$ polynomial can be done via Horner's method. This uses one sum register, one double sized multiplication output register and one counter register. The multiplication output register has twice the size since before taking the remainder, the full product as a polynomial has to be stored. Since the irreducible polynomial has been precomputed, and we can compute remainders in place, we can evaluate a polynomial in additional space $3r_n + \lceil \log 2e \rceil + O(\log r_n)$. 
            Iterating over all possible solutions uses an additional $r_n $ space. This procedure can recycle the space used in step 11. We use an additional register size $\lceil \log 2e \rceil$ to find the $i$th root. Counting the space used to store the irreducible polynomial means that this costs space $(4e + 5)r_n + O(\log^2 r_n + \log e)$.
        \item We add a (double-sized) multiplication output register for multiplication, a register to maintain the product, and another set of registers to iterate over all possible roots.  Iterating over all roots not equal to $x^{-1}$ allows us to then compute $\alpha_{x_i}$. We then take the multiplicative inverse using Lemma \ref{lemma:mult_inverse}. Overall, this uses space $(4e + 6) + O(\log r_n)$ by reusing registers.
        \item Reusing the space from steps 11 and 12 we can compute the value of the error by multiplying $\omega(x^{-1})$ by $\alpha_{x_i}^{-1}$.
    \end{enumerate}
    This covers all steps of Algorithm \ref{alg:decode} with significant space costs. We ignore $O(\log \delta)$ space terms here, since these are all $O(\log n)$. This adds us to $(6 + 4e) r_n + O(\log^2 r_n)$ space.
\end{proof}



\end{document}